\documentclass[11pt]{article}
\usepackage[utf8]{inputenc}
\usepackage{amsmath, amsthm, amssymb, mathtools}
\usepackage{amsfonts}
\usepackage{amssymb}
\usepackage{bbm}
\usepackage{graphicx}
\usepackage{caption}
\usepackage{fullpage}
\usepackage{tikz}
\usepackage{microtype}
\usepackage{hyperref}
\usepackage{thmtools}
\usepackage{thm-restate}
\usepackage[colorinlistoftodos,prependcaption,textsize=small]{todonotes}
\usepackage[linesnumbered,ruled,procnumbered]{algorithm2e}
\usepackage{algorithmicx}

\hypersetup{
    colorlinks=true,
    linkcolor=blue,
    filecolor=magenta,      
    urlcolor=cyan,
}

\newtheorem{theorem}{Theorem}
\newtheorem{claim}[theorem]{Claim}

\newtheorem{lemma}{Lemma}
\newtheorem{assumption}{Assumption}

\newcommand{\norm}[1]{\left\|#1\right\|}
\newcommand{\pnorm}[2]{\bigl\|#2\bigr\|_#1}

\newcommand{\rank}{\mathrm{rank}}
\newcommand{\polylog}{\mathrm{polylog}}
\newcommand{\poly}{\mathrm{poly}}
\newcommand{\Diag}{\mathrm{diag}}
\newcommand{\sspan}{\mathrm{span}}

\newcommand{\bbr}{\mathbb{R}}
\newcommand{\bbc}{\mathbb{C}}
\newcommand{\E}{\mathbb{E}}
\newcommand{\Var}{\mathrm{Var}}

\begin{document}
\title{Quantum-inspired sublinear classical algorithms for solving low-rank linear systems}
\author{Nai-Hui Chia\thanks{University of Texas at Austin, Department of Computer Science.} \and  Han-Hsuan Lin$^{*}$ \and Chunhao Wang$^{*}$}
\date{\empty}
\maketitle

\begin{abstract}
  We present classical sublinear-time algorithms for solving low-rank linear systems of equations. Our algorithms are inspired by the HHL quantum algorithm~\cite{HHL09} for solving linear systems and the recent breakthrough by Tang~\cite{T18} of dequantizing the quantum algorithm for recommendation systems. Let $A \in \mathbb{C}^{m \times n}$ be a rank-$k$ matrix, and $b \in \mathbb{C}^m$ be a vector. We present two algorithms: a ``sampling'' algorithm that provides a sample from $A^{-1}b$ and a ``query'' algorithm that outputs an estimate of an entry of $A^{-1}b$, where $A^{-1}$ denotes the Moore-Penrose pseudo-inverse. Both of our algorithms have query and time complexity $O(\mathrm{poly}(k, \kappa, \|A\|_F, 1/\epsilon)\,\mathrm{polylog}(m, n))$, where $\kappa$ is the condition number of $A$ and $\epsilon$ is the precision parameter. Note that the algorithms we consider are sublinear time, so they cannot write and read the whole matrix or vectors. In this paper, we assume that $A$ and $b$ come with well-known low-overhead data structures such that entries of $A$ and $b$ can be sampled according to some natural probability distributions. Alternatively, when $A$ is positive semidefinite, our algorithms can be adapted so that the sampling assumption on $b$ is not required.
\end{abstract}

\section{Introduction}

The problem of solving linear systems of equations plays a fundamental role in many fields. In this problem, we are given a matrix $A \in \bbc^{m \times n}$ and a vector $b \in \bbc^m$. The objective is to find a vector $x \in \bbc^n$ such that $Ax = b$. This problem reduces to finding the inverse (or pseudo-inverse) of $A$, which can be solved by singular value decomposition on $A$ using $O(\min\{mn^2, nm^2\})$ operations. In recent decades, many algorithms have been proposed to reduce the time complexity for different special classes of matrices. For general square matrices, the best known algorithm is based on matrix multiplication with running time $O(n^{\omega})$ for $\omega < 2.373$ (see \cite{LeGall14}). When $A$ is $d$-sparse (i.e., there are at most $d$ nonzero entries in each row/column), Spielman \cite{Spielman10} gave an algorithm with running time $O(\max\{dm, dn\})$. More specially, when $A \in \bbc^{n \times n}$ is $d$-sparse and symmetric diagonally dominant, Spielman and Teng \cite{ST04} gave the first near-linear time algorithm with running time $O(d\log^{O(1)}\log(1/\epsilon)\, n)$, where $\epsilon$ is the precision parameter. Cohen \textit{et al.}~\cite{CKM+14} improved the running time to $O(d\sqrt{\log n}(\log\log n)^{O(1)}\log(1/\epsilon))$. Recently, a notable breakthrough for symmetric diagonally dominated matrices has been achieved by Andoni, Krauthgamer, and Pogrow~\cite{AKP18}, who demonstrated that computing one entry of $A^{-1}b$ can be done in sublinear time. 

In terms of quantum algorithms, in 2009, Harrow, Hassidim, and Lloyd~\cite{HHL09} gave a quantum algorithm (referred to as the HHL algorithm) that solves linear systems of equations in time $O(\polylog(n))$ for $n$-by-$n$ sparse matrices. It is also shown in \cite{HHL09} that inverting sparse matrices is \textbf{BQP}-complete. As a result, it is very likely that the HHL algorithm does hold an exponential speedup against classical ones. Following \cite{HHL09}, a series of quantum algorithms with better dependences on parameters has been proposed \cite{Ambainis12,CKS17}, and a survey of these quantum algorithms and their applications can be found in \cite{DHM+18}. 

Note that when solving linear system of equations, just writing down the answer $x$ requires time $n$, so obviously it is impossible to output $x$ in logarithmic time. Instead, the HHL algorithm outputs a quantum state (a normalized complex vector) whose amplitudes (the entries of this vector) are proportional to the corresponding entries of the answer $x$. Thus the classical information one can obtain is \emph{samples} from $\{1, \ldots, n\}$ distributed according to probability $|x(i)|^2/\norm{x}^2$ and the value $x^\dag M x$ for Hermitian matrices $M$ that can be efficiently implemented as a quantum operator.

The HHL algorithm inspired a series of logarithmic time quantum algorithms for various machine learning related problems, including but not limited to finding least square approximation~\cite{WBL12}, principal component analysis~\cite{LMR14}, support vector machine~\cite{RML14}, and recommendation systems~\cite{KP16}. Like the HHL algorithm, these logarithmic time algorithms do not give a full description of the answer, but give a quantum state from which one can sample. Noticing that a logarithmic time \emph{classical} algorithm might also be able to sample from the answer, Tang~\cite{T18} in their recent breakthrough work presented a logarithmic time {classical} algorithm for recommendation systems by employing the techniques of efficient low-rank approximation by Frieze, Kannan, and Vempala~\cite{FKV04}, showing that~\cite{KP16} does not give an exponential quantum speedup. Recently, Tang pointed out that their techniques also worked for other machine learning problems such as principal component analysis and supervised clustering~\cite{T18-2}.

Our work is inspired by Frieze \textit{et al.}~\cite{FKV04} and Tang~\cite{T18}. We solve the linear system of equations $Ax=b$ for a \emph{low-rank} matrix $A$ and a vector $b$ that have some natural sampling assumptions. Like the HHL algorithm, we do not output the whole vector $x$. Instead, we give logarithmic-time algorithms to sample from $x$ with probability $|x(i)|^2/\norm{x}^2$ and to estimate an entry $x(i)$. It is also possible to efficiently estimate $x^{\dag}Mx$ using our approach. More specifically, our algorithms depend on the ability to sample a row index of $A$ according to the norms of row vectors, the ability to sample an entry in each row according to the absolute value of each entry (see Assumption~\ref{asmp:sample-m}), and the ability to sample an entry from $b$ according to the absolute values of its entries (see Assumption~\ref{asmp:sample-v}). We discuss the details of these sampling assumptions in Section~\ref{sec:assump}.

Roughly speaking, our algorithms are based on the sampling techniques in~\cite{FKV04} (to get a small submatrix) and~\cite{T18} (to sample from a vector formed by a matrix-vector multiplication). To be able to manipulate (e.g., invert) the singular values, we introduce a \emph{new succinct description} of the resulting approximation matrix based on the small submatrix, which might be of independent interest and could potentially lead to other applications. In addition, we also propose a method to efficiently estimate $x^{\dag}Ay$ given the sampling access to $x$ and $y$, which extends the sampling tools provided by Tang~\cite{T18}.

\subsection{Notations, problem definition, and main results}
In this paper, we use $\norm{M}$ to denote the spectral norm and use $\|M\|_F$ to denote the Frobenius norm of the matrix $M$. We use $M(i, \cdot)$ to denote the $i$-th row of $M$, which is a row vector, and use $M(\cdot, j)$ to denote the $j$-th column of $M$, which is a column vector. The $(i,j)$-entry of $M$ is denoted by $M(i,j)$. For a vector $v$, we use $v(i)$ to denote the $i$-th entry of $v$. The complex-conjugate transpose of a matrix $M$ (and a vector $v$, respectively) is denoted by $M^{\dag}$ (and $v^{\dag}$, respectively). For a nonzero vector $x\in \mathbb{C}^n$, we denote by $\mathcal{D}_x$ the probability distribution on $\{1, \ldots, n\}$ where the probability that $i$ is chosen is defined as $\mathcal{D}_x(i) = |x(i)|^2/\norm{x}^2$ for all $i \in \{1, \ldots, n\}$. A sample from $\mathcal{D}_x$ is often referred to as a sample from $x$. Given two vectors $x, y \in \bbc^n$, the total variation distance between $\mathcal{D}_x$ and $\mathcal{D}_y$, denoted by $\|\mathcal{D}_x, \mathcal{D}_y\|_{TV}$, is defined as $\|\mathcal{D}_x,\mathcal{D}_y\|_{TV} := \frac{1}{2}\sum_{i=1}^n |\mathcal{D}_x(i) - \mathcal{D}_y(i)|.$

Let $A \in \bbc^{m \times n}$ be a matrix with $\rank(A) = k$. The usual definition of the condition number (as $\norm{A}/\|A^{-1}\|$) is not well-defined for a singular matrix. Here, we slightly change the definition of the \emph{condition number} of a singular matrix $A$, denoted by $\kappa$ as $\kappa = \norm{A}/\sigma_{\min}(A)$, where $\sigma_{\min}(A)$ is the minimum \emph{nonzero} singular value of $A$. We use $A^{-1}$ to denote the \emph{Moore-Penrose pseudo-inverse} of $A$, i.e., if $A$ has the singular value decomposition $A = \sum_{i=1}^k\sigma_iu_iv_i^{\dag}$, then $A^{-1} := \sum_{i=1}^k \sigma_i^{-1}v_iu_i^{\dag}$. In the problem of solving linear systems of equations, the objective is to query to and sample from $A^{-1}b$ for a given vector $b \in \bbc^m$. We summarize our main results in the following theorems. The first theorem asserts the ability to query an entry to the answer $A^{-1}b$.

\begin{restatable}{theorem}{mainquery}
  \label{thm:main1}
  Let $A \in \bbc^{m \times n}$ be a matrix that has the sampling access as in Assumption~\ref{asmp:sample-m}, and $b \in \bbc^m$ be a vector with the sampling access as in Assumption~\ref{asmp:sample-v}. Let $\kappa$ be the condition number of $A$. There exists an algorithm to approximate $(A^{-1}b)(i)$ for a given index $i \in \{1, \ldots, n\}$ with additive error $\epsilon$ and success probability $1-\delta$ by using 
  \begin{align}
    O\left(\poly\left(k, \kappa, \pnorm{F}{A}, \frac{1}{\epsilon}\right)\, \polylog(m, n)\, \log\left(\frac{1}{\delta}\right)\right)
  \end{align}
  queries and time.
  
\end{restatable}

The next theorem addresses the ability to sample from the answer $A^{-1}b$.
\begin{restatable}{theorem}{mainsample}
  \label{thm:main2}
  Let $A \in \bbc^{m \times n}$ be a matrix that has the sampling access as in Assumption~\ref{asmp:sample-m}, and $b \in \bbc^m$ be a vector with the sampling access as in Assumption~\ref{asmp:sample-v}. Let $\kappa$ be the condition number of $A$. There exists an algorithm to sample from a distribution which is $\epsilon$-close to  $\mathcal{D}_{A^{-1}b}$ in terms of total variation distance with success probability $1-\delta$ by using 
  \begin{align}
    O\left(\poly\left(k, \kappa, \pnorm{F}{A}, \frac{1}{\epsilon}\right)\, \polylog(m, n)\, \log\left(\frac{1}{\delta}\right)\right)
  \end{align}
  queries and time.
\end{restatable}

\noindent
\textbf{Remarks:}
\begin{enumerate}
  \item When $b$ is not entirely in the left-singular vector space of $A$, elements of $A^{-1}b$ might be so small that the additive error by sampling and approximation dominates the value of the Algorithms' outcomes. However, if $b$ has little or zero overlap with this space, we can detect this case by evaluating the inner product $b^{\dag}AA^{-1}b$. We leave the discussion to Subsection~\ref{subsec:b-not-in-a}.
  \item Although our algorithms have poly-logarithmic time complexity, there are large constant factors and exponents in the polynomial (see the proofs of the main theorems in Section~\ref{sec:main-alg}). We expect that those large constant factors and exponents are just consequences of our analysis. In practice, we expect the number of samples needed is much smaller than the upper bound we give.
  \item Let $M \in \bbc^{n \times n}$ be a matrix (with no sampling assumption on it). Then $x^{\dag}Mx$ can be estimated efficiently (using Lemma~\ref{lem:samp_4}).
  \item When $A$ is positive semidefinite, our algorithms can be adapted so that no sampling assumption on $b$ (Assumption~\ref{asmp:sample-v}) is needed. This is discussed in Subsection~\ref{subsec:no-sampling-b}. 
\end{enumerate}

\subsection{Outline of the algorithms}
Our algorithms consist in two main steps. The first step is to sample a small submatrix from $A$ and compute the singular values and singular vectors of this small submatrix. Then, $A^{\dag}A$ can be approximately reconstructed from these singular values and singular vectors of this small submatrix. However, instead of reconstructing $A^{\dag}A$, we perform the second step: sampling or query to $A^{-1}b$ based on the properties of these singular values and singular vectors of this small submatrix. In the following, we describe the intuitions of these two steps.

\paragraph{Subsampling from $A$.} The subsampling methods is same as that of \cite{FKV04}; however, we use a new succinct description of the approximation of $A^\dag A$ which allows us to manipulate the singular values of $A$.

The intuition of the subsampling and succinct description is as follows. For a matrix $A \in \bbc^{m \times n}$ of rank $k$, first we sample $k$ rows from $A$ to obtain a submatrix $S \in \bbc^{k \times n}$. Then we sample $k$ columns from $S$ to obtain a submatrix $W \in \bbc^{k \times k}$. We compute the singular values $\hat{\sigma}_1, \ldots, \hat{\sigma}_k$ of $W$ and their corresponding left singular vectors $\hat{u}_1, \ldots, \hat{u}_k$. Let $V \in \bbc^{n \times k}$ be the matrix formed by the column vectors $\frac{S^{\dag}}{\hat{\sigma}_i}\hat{u}_i$, and let $D \in \bbr^{k \times k}$ be the diagonal matrix with diagonal entries $\hat{\sigma}_1, \ldots, \hat{\sigma}_k$. We define $\widehat{A}^{\sim2} = VD^2V^{\dag}$, which is our approximation of $A^{\dag}A$. The fact that $\widehat{A}^{\sim2}$ is close to $A^{\dag}A$ is established in Lemma~\ref{lemma:subsample}. Now, define $\widehat{A}^{\sim-2} = VD^{-2}V^{\dag}$, and the distance between $\widehat{A}^{\sim-2}$ and $(A^{\dag}A)^{-1}$, is bounded by Lemma~\ref{lem:nx}. In the next step, we work on $\widehat{A}^{\sim-2}$. Note that the data that our algorithms actually store are only the singular values and left singular vectors of $W$, as well as the indices we sampled to form $S$ and $W$. We do not store $V$, $S$, $\widehat{A}^{\sim2}$, or $\widehat{A}^{\sim-2}$ , as those matrices are too large to be efficiently stored. What we do provide is sampling and query access to $V$ from our stored data. The presence of $\widehat{A}^{\sim2}$ and $\widehat{A}^{\sim-2}$ is just for analysis.

\paragraph{Sampling from and querying to $\widehat{A}^{\sim-2}A^{\dag}b$.} Note that $A^{-1} = (A^{\dag}A)^{-1}A^{\dag} \approx \widehat{A}^{\sim-2}A^{\dag}$, so it suffices to work on $\widehat{A}^{\sim-2}A^{\dag}b$. In this step, we develop a new sampling tool: given query and sample access to vectors $u \in \bbc^n$ and $v \in \bbc^m$ and query access to a matrix $A \in \bbc^{n \times m}$, we give a procedure to estimate $u^{\dag}Av$ in sublinear time (Lemma~\ref{lem:samp_4}). Using this procedure $k$ times, we obtain the $k \times 1$ vector $V^{\dag}A^{\dag}b$. Then $D^{-2}V^{\dag}A^{\dag}b$ can be computed efficiently, and by using the sampling tools in~\cite{T18}, the vector $V(D^{-2}V^{\dag}A^{\dag}b)$ can be sampled and queried in sublinear time.

\section{Sampling assumptions and data structure}
\label{sec:assump}

We are interested in developing sublinear-time algorithms for linear systems, so we need to concern ourselves with the way the input matrix and vector are given. Obviously, one cannot load the full matrix and vector into the memory since parsing them requires at least linear time. In this paper, we assume the matrix and vector can be sampled according to some natural probability distributions that arise in many applications in machine learning (see~\cite{KP16,T18,T18-2}, and also discussed in~\cite{FKV04}).

We first present the sampling assumptions for a matrix. Intuitively speaking, we assume that we can sample a row index according to the norms of its row vectors, and for each row, we can sample an entry according to the absolute values of the entries in that row.
\begin{assumption}
  \label{asmp:sample-m}
  Let $M \in \bbc^{m \times n}$ be a matrix. Then, the following conditions hold
  \begin{enumerate}
    \item We can sample a row index $i \in \{1, \ldots, m\}$ of $M$ where the probability of row $i$ being chosen is
      \begin{align}
        P_i = \frac{\norm{M(i, \cdot)}^2}{\pnorm{F}{M}^2}.
      \end{align}
    \item For all $i \in \{1, \ldots, m\}$, we can sample an index $j \in \{1, \ldots, n\}$ according to $\mathcal{D}_{M(i,\cdot)}$, i.e., the probability of $j$ being chosen is
      \begin{align}
        \mathcal{D}_{M(i, \cdot)}(j) = \frac{|M(i,j)|^2}{\norm{M(i,\cdot)}^2}.
      \end{align}
  \end{enumerate}
\end{assumption}

Similarly, for a vector, we assume that we can sample an entry according to the absolute values of its entries.
\begin{assumption}
  \label{asmp:sample-v}
  Let $v \in \bbc^{n}$ be a vector. We can sample an index $i \in \{1, \ldots, n\}$ according to $\mathcal{D}_v$, i.e., the probability of $i$ being chosen is
      \begin{align}
        \mathcal{D}_v(i) = \frac{|v(i)|^2}{\norm{v}^2}.
      \end{align}
\end{assumption}

In fact, these assumptions are empirical. Frieze \textit{et al.}~\cite{FKV04} used the similar assumptions to develop sublinear algorithms for finding low-rank approximation. As pointed out by~\cite{KP16} and also used in~\cite{T18,T18-2}, there exists low-overhead data structures that fulfill these sampling assumptions. More precisely, we summarize the existence of such data structures for Assumption~\ref{asmp:sample-m} as follows.

\begin{theorem}[\cite{KP16}]
  \label{thm:data-m}
  Given a matrix $A\in \bbc^{m \times n}$ with $s$ non-zero entries, there exists a data structure storing $A$ in space $O(s\log^2n)$, which supports the following operations:  
  \begin{itemize}
    \item Reading and writing $A(i,j)$ in $O(\log^2(mn))$ time.
    \item Evaluating $\norm{A(i, \cdot)}$ in $O(\log^2 m)$ time.
    \item Evaluating $\pnorm{F}{A}^2$ in $O(1)$ time. 
    \item Sampling a row index of $A$ according to statement 1 of Assumption~\ref{asmp:sample-m} in $O(\log^2(mn))$ time.   
    \item For each row, sampling an index $j$ according to statement 2 of Assumption~\ref{asmp:sample-m} in $O(\log^2(mn))$ time. 
  \end{itemize}
\end{theorem} 

As a special case of the Theorem~\ref{thm:data-m} (where the matrix contains a single row), the existence of such data structures for Assumption~\ref{asmp:sample-v} is summarized as follows.
\begin{theorem}\label{thm:data-v}
  Given a vector $v\in \bbc^{n}$ with $s$ non-zero entries, there exists a data structure storing $v$ in space $O(s\log^2n)$, which supports the following operations:  
  \begin{itemize}
    \item Reading and writing $v(i)$ in $O(\log^2n)$ time.
    \item Evaluating $\norm{v}^2$ in $O(1)$ time. 
    \item Sampling an index of $v$ according to Assumption~\ref{asmp:sample-v} in $O(\log^2n)$ time.   
  \end{itemize}
\end{theorem} 

For the details of the proof for Theorem~\ref{thm:data-m}, one may refer to~\cite{KP16}. Here, we give the intuition of the data structure as follows. It suffices to show how to sample from a single vector. As demonstrated in Fig.~\ref{fig:ds}, we use a binary tree to store the data of a vector: the square of the absolute value of each entry, along with its original value (which is a complex number) are stored in the leaf nodes; each internal node contains the sum of the values of its two immediate children. In this way, the root note contains the norm of this vector. To sample an index and to query an entry from this vector, logarithmic steps suffice.

\begin{figure}[ht]
  \centering
  \includegraphics[width=0.8\textwidth]{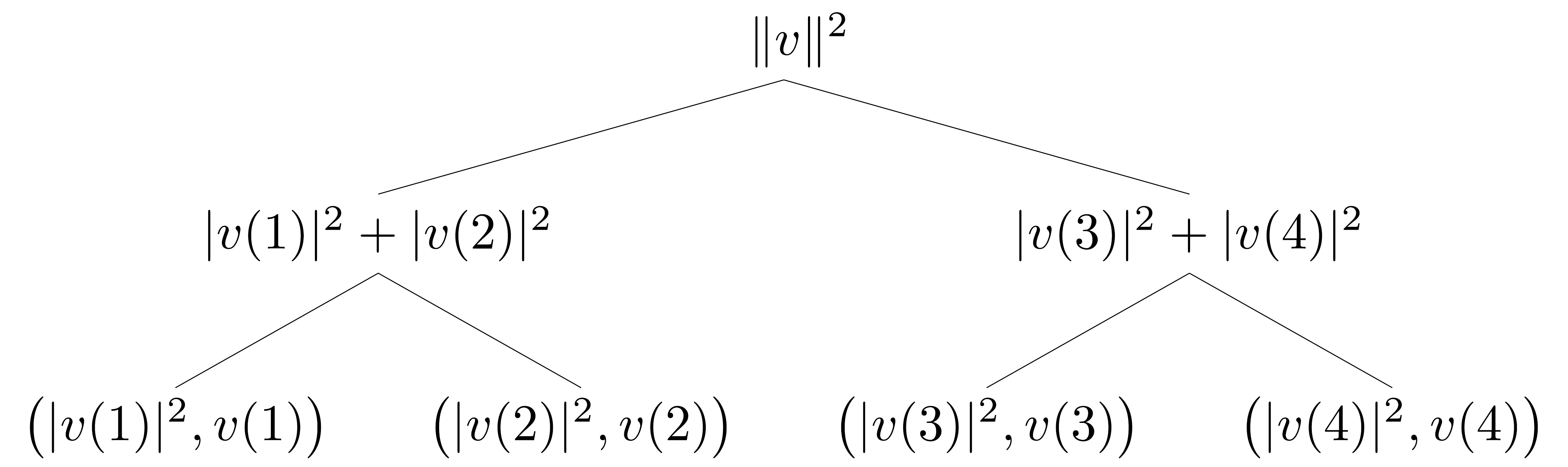}
  \caption{Illustration of a data structure that allows for sampling access to a vector $v \in \bbc^4$. \label{fig:ds}}
\end{figure}

\section{Technical lemmas}
\label{sec:tech-lemmas}
In this section, we present the technical lemmas that will be used to prove the main results. We first prove some tools to bound the distance between matrices, which will be used in Section~\ref{sec:samping} to show the succinct description obtained from the sampled submatrix is a good approximation to the desired matrix. Then, we show how to sample from a vector obtained from a matrix-vector multiplication and how to estimate $x^{\dag}Ay$ in logarithmic time, provided we are given the sampling ability as in Assumptions~\ref{asmp:sample-m} and \ref{asmp:sample-v}. These sampling techniques will be used to prove the main theorem in Section~\ref{sec:main-alg}. Some of these lemmas might be of independent interest.

\subsection{Distance between matrices}
The following lemma shows that if the distance between the squares of two positive semidefinite matrices is small, then the distance between these two matrices should also be small.
\begin{lemma}
  \label{lemma:x2-y2}
  Let $X, Y \in \bbc^{n\times n}$ be positive semidefinite matrices with $\max\{\rank(X), \rank(Y)\} = k$. It holds that $\pnorm{F}{X - Y} \leq (2k)^{1/4}\pnorm{F}{X^2-Y^2}^{1/2}$.
\end{lemma}

\begin{proof}
  Since $X$ and $Y$ are positive semidefinite, we can diagonalize them. As the maximum of their ranks is $k$, the numbers of their nonzero eigenvalues are no larger than $k$.  Consider the space spanned by the union of their eigenvectors. The space is at most $2k$-dimensional, and $X$ and $Y$ are both trivially in this space.  
  
  Let $\Delta = X-Y$, and let $\{\lambda_1, \ldots, \lambda_{2k}\}$ and $\{v_1, \ldots, v_{2k}\}$ be the eigenvalues and eigenvectors of $\Delta$, i.e. $\Delta v_i = \lambda_i v_i$, for all $i \in \{1, \ldots, 2k\}$. Because $\Delta$ is symmetric, $\lambda_i$'s are real. Together with fact that $v_i^{\dag} (X-Y)v_i=\lambda_i$, we know either $v_i^{\dag} (X-Y)v_i=|\lambda_i|$ or $v_i^{\dag} (Y-X)v_i=|\lambda_i|$ is true. Since $X$ and $Y$ are positive semidefinite, we have
  \begin{align}
    \label{eq:A+B}
    \left|v_i^{\dag} (X+Y)v_i\right| \geq|\lambda_i|, \text{ for all } i \in \{1, \ldots, 2k\}.
  \end{align}
  We now calculate $\pnorm{F}{X^2-Y^2}$ as follows.
  \begin{align}
    \pnorm{F}{X^2-Y^2}^2 &=  \pnorm{F}{(Y+\Delta)^2-Y^2}  \nonumber \\
                      &=\pnorm{F}{Y \Delta +\Delta X}  \nonumber \\
                      &=\sum_{i,j=1}^{2k} |v_i^{\dag} (Y\Delta +\Delta X) v_j|^2 \nonumber \\
                      & \ge \sum_{i=1}^{2k} |v_i^{\dag} (Y\Delta +\Delta X) v_i|^2 \nonumber \\
                      &=\sum_{i=1}^{2k}  \lambda_i^2 |v_i^{\dag}(Y+X)v_i|^2 \nonumber \\
                      &\ge \sum_{i=1}^{2k} \lambda_i^4 =\sum_{i=1}^{2k} (\lambda_i^2)^2 \left(\frac{1}{2k}\sum_{i=1}^{2k} 1^2\right) \nonumber \\
                      &\geq \frac{1}{2k}\left(\sum_i \lambda_i^2\right)^2 =\frac{1}{2k} \left( \pnorm{F}{X-Y}^2\right)^2,
  \end{align}
  where the second inequality follows from Eq.~\eqref{eq:A+B}, and the third inequality follows from the Cauchy-Schwartz inequality. Multiplying both side by $2k$ and taking the fourth-root, we get the desired inequality.
\end{proof}

The next lemma asserts that if the distance between two positive semidefinite matrices is small, then the distance between their pseudo-inverses is also small.
\begin{lemma}
  \label{lemma:x-1-y-1}
  Let $X, Y \in \bbc^{n\times n}$ be positive semidefinite matrices with $\max\{\rank(X), \rank(Y)\} = k$. Let $\sigma_{\min}=\min\{\sigma_{\min}(X),\sigma_{\min}(Y)\}$, where $\sigma_{\min}(\cdot)$ is the minimum nonzero singular value of a matrix. It holds that 
  \begin{align}
    \pnorm{F}{X^{-1} - Y^{-1}} \leq \frac{3 \pnorm{F}{X-Y}}{\sigma_{\min}^2}.
  \end{align}
\end{lemma}

\begin{proof}
  As in the beginning of the proof of Lemma~\ref{lemma:x2-y2}, we consider the space spanned by the union of the eigenvectors of $X$ and $Y$, which is at most $2k$-dimensional, and $X$ and $Y$ are both in this space. 
  
  We start by diagonalizing $X$ and $Y$ as follows.
  \begin{align}
    X= \sum_{i=1}^{2k} \sigma_i u_i u_i^{\dag}, \quad \text{ and } \quad Y= \sum_{i=1}^{2k} \sigma'_i u'_i u'^{\dag}_i.
  \end{align}
  Let $r=\rank(X)$ and $r'=\rank(Y)$, we define the projectors associated with their nontrivial eigenvectors:
  \begin{align}
    \Pi_X = \sum_{i=1}^{r} u_i u_i^{\dag}, \quad \text{ and } \quad \Pi_Y = \sum_{i=1}^{r'} u'_i u'^{\dag}_i.
  \end{align}
  
  For all $i \in \{1, \ldots, r\}$, define the unnormalized vectors $v_i= \frac{Y u_i}{\sigma_i}$. We have $ v_i \in \sspan\{u_i'\}$ for $i\in \{1, \ldots, r\}$ and 
  \begin{align}
    \sum_{i=1}^r\norm{\sigma_iu_i^{\dag} - u_i^{\dag}Y}^2 = \sum_{i=1}^r\norm{u_i^{\dag}(X - Y)}^2 \leq \pnorm{F}{X-Y}^2.
  \end{align}
  It further implies that,
  \begin{align}
    \label{eq:uv-distance}
    \sum_{i=1}^r\norm{u_i^{\dag} - v_i^{\dag}}^2 = \sum_{i=1}^r\frac{1}{\sigma_i^2}\norm{\sigma_iu_i^{\dag} - \sigma_iv_i^{\dag}}^2 = \sum_{i=1}^r\frac{1}{\sigma_i^2}\norm{\sigma_iu_i^{\dag} - u_i^{\dag}Y}^2 \leq \frac{\pnorm{F}{X-Y}^2}{\sigma_{\min}^2}.
  \end{align}

  Now, We cut $\pnorm{F}{X^{-1}-Y^{-1}}$ into three parts as follows:
  \begin{align}
    \pnorm{F}{X^{-1}-Y^{-1}} &=\pnorm{F}{X^{-1}-X^{-1}\Pi_Y+X^{-1}\Pi_Y-\Pi_XY^{-1} +\Pi_X Y^{-1}-Y^{-1}} \nonumber \\
                          &\leq  \pnorm{F}{X^{-1}-X^{-1}\Pi_Y}+\pnorm{F}{X^{-1}\Pi_Y-\Pi_XY^{-1}} +\pnorm{F}{\Pi_X Y^{-1}-Y^{-1}} \nonumber \\
                          &\leq  \pnorm{F}{X^{-1}\Pi_X(I-\Pi_Y)}+\pnorm{F}{X^{-1}Y Y^{-1} -X^{-1} XY^{-1}} +\pnorm{F}{(\Pi_X-I)\Pi_Y Y^{-1}} \nonumber \\
                          &\leq   \norm{X^{-1}}\pnorm{F}{\Pi_X(I-\Pi_Y)}+\norm{X^{-1}}\pnorm{F}{Y - X}\norm{Y^{-1}} +\pnorm{F}{(\Pi_X-I)\Pi_Y }\norm{Y^{-1}} \nonumber \\
                          &\leq \frac{1}{\sigma_{\min}}\pnorm{F}{\Pi_X(I-\Pi_Y)}+\frac{1}{\sigma_{\min}^2}\pnorm{F}{Y-X} +\frac{1}{\sigma_{\min}}\pnorm{F}{(\Pi_X-I)\Pi_Y } \label{eq:3parts}
  \end{align}
  We then bound $\pnorm{F}{\Pi_X(I-\Pi_Y)}$ as follows:
  \begin{align}
    \pnorm{F}{\Pi_X(I-\Pi_Y)}^2 &=\left\|\Pi_X\left(I- \sum_{j=1}^{r} u_j v_j^{\dag} \right)(I-\Pi_Y)\right\|_F^2 \nonumber \\
                                & \leq \left\|\Pi_X \left(I- \sum_{j=1}^{r} u_j v_j^{\dag} \right)\right\|_F^2 \norm{(I-\Pi_Y)}^2\nonumber \\
                                &= \sum_{i=1}^{2k} \norm{u_i^{\dag} \Pi_X \left(I- \sum_{j=1}^{r} u_j v_j^{\dag} \right)}^2 \nonumber \\
                                &= \sum_{i=1}^{r} \norm{u_i^{\dag}-  v_i^{\dag}}^2 \leq \frac{\pnorm{F}{Y-X}^2}{\sigma_{\min}^2}, \label{eq:pxpy}
  \end{align}
  where the first equality follows from that fact that $v_i \in \sspan\{u_i'\}$ and the last inequality follows from Eq.~\eqref{eq:uv-distance}. By the symmetry between $X$ and $Y$, we also have 
   \begin{align}
     \pnorm{F}{(\Pi_X-1)\Pi_Y} \leq \frac{\pnorm{F}{Y-X}}{\sigma_{\min}}. \label{eq:pypx}
   \end{align}
  Putting Eqns.~\eqref{eq:pxpy} \eqref{eq:pypx} back into Eq.~\eqref{eq:3parts}, we have
  \begin{align}
    \pnorm{F}{X^{-1}-Y^{-1}} \leq \frac{3\pnorm{F}{X-Y}}{\sigma_{\min}^2}.
  \end{align}
\end{proof}

\subsection{Sampling techniques}\label{sec:samp_tech}
The following lemma from Frieze \textit{et al.}~\cite{FKV04} was originally for real matrices, but it is easy to generalize to complex matrices. We do not repeat the proof here.
\begin{lemma}[\cite{FKV04}]
  \label{lemma:subsample}
  Let $M \in \bbc^{m \times n}$ be a matrix. Independently sample $p$ row indices $i_1, \ldots, i_p$ according to the probability distribution $\{\norm{M(1,\cdot)}^2/\pnorm{F}{M}^2, \ldots, \norm{M(m,\cdot)}^2/\pnorm{F}{M}^2\}$. Let $N \in \bbc^{p \times n}$ be the normalized submatrix of $M$ with
  \begin{align}
    N(i_t, \cdot) = \frac{M(i_t, \cdot)}{\sqrt{p\norm{M(i_t,\cdot)}^2}/\pnorm{F}{M}^2},
  \end{align}
  for $t \in \{1, \ldots, p\}$. Then, for all $\theta < 0$, it holds that
  \begin{align}
    \Pr\left(\pnorm{F}{M^TM - N^TN} \geq \theta \pnorm{F}{M}^2\right) \leq \frac{1}{\theta^2p}.
  \end{align}
\end{lemma}

With the data structure in Theorem~\ref{thm:data-m}, one can estimate the inner product and of two vectors and sample the vector resulted from a matrix-vector multiplication. 

\begin{restatable}[\cite{T18}]{lemma}{samplinnerproduct}
  \label{lem:samp_1}
  Let $x,y\in \mathbb{C}^n$. Given query access to $x$ and $y$, the ability to sample from $\mathcal{D}_x$, and the knowledge of $\norm{x}$, one can approximate $\langle x,y\rangle$ to additive error $\epsilon\|x\|\|y\|$ with at least $1-\delta$ probability using $O(\frac{1}{\epsilon^2}\log \frac{1}{\delta})$ queries and samples and the same time complexity. 
\end{restatable}

\begin{restatable}[\cite{T18}]{lemma}{samplingmv}
  \label{lem:samp_2}
  Let $M\in \mathbb{C}^{n\times k}$ and $v\in \mathbb{C}^k$. Given sampling access to $M$ as in Assumption~\ref{asmp:sample-m}, one can output a sample from the vector $Mv$ with probability $9/10$ in $O(k^2 C(M,v))$ query and time complexity, where 
\[
	C(M,v) := \frac{\sum_{j=1}^{k} \parallel v_jM(\cdot,j)\parallel^2}{\parallel Mv\parallel^2}.
\]
\end{restatable}

\begin{restatable}[\cite{T18}]{lemma}{samplingalmostortho}
  \label{lem:samp_3}
  Let $M\in \mathbb{C}^{n\times k}$ and $ v\in \mathbb{C}^k$. If there exists an isometry $U \in \bbc^{n \times k}$ whose column vectors span the column space of $M$ such that $\pnorm{F}{M - U} \leq \alpha$, then 
  one can sampling from a distribution which is $(\alpha+O(\alpha^2))$-close to $\mathcal{D}_{Mv}$ in $O(k^2(1+O(\alpha)))$ expected query and time complexity.   
\end{restatable}

Note that when the vectors and the matrices are real, Tang~\cite{T18} has proven Lemma~\ref{lem:samp_1}, Lemma~\ref{lem:samp_2}, and Lemma~\ref{lem:samp_3}. Their proofs can be extended to complex matrices and vectors.   

To solve the linear system, we prove the following lemma. 
\begin{lemma}\label{lem:samp_4}
  Let $x\in \mathbb{C}^m$, $y\in \mathbb{C}^n$, and $A\in \mathbb{C}^{m \times n}$. Given query access to $x$, $A$, and $y$, the ability to sample from $\mathcal{D}_x$ and $\mathcal{D}_y$, and the knowledge of $\norm{x}$ and $\norm{y}$, one can approximate $x^{\dag}Ay$ to additive error $\epsilon$ with at least $1-\delta$ success probability using $O(\frac{\|x\|\|y\|\|A\|_F}{\epsilon^2}\log \frac{1}{\delta})$ queries and samples and the same time complexity. 
\end{lemma}
\begin{proof}
Define a random variable $Z$ as follows: 
\[
    Z = \frac{\norm{x}^2\norm{y}^2A(i,j)}{x(i)y^*(j)} \mbox{ with probability } \frac{|x(i)|^2|y(j)|^2}{\|x\|^2\|y\|^2}. 
\]
The expected value and variance of $Z$ are 
\[
    \E[Z] = y^{\dag} Ax, \mbox{ and } \Var[|Z|] = \norm{x}^2\norm{y}^2\norm{A}_F^2. 
\]

Then, we prove this lemma by the technique of median of means. Given $pq$ many samples of $Z$'s, we divide these samples into $p$ groups $S_1,\dots, S_p$. We let $Y_i = \frac{\sum_{j=1}^q Z_j}{p}$ be the mean of $S_i$ and $\tilde{Y}$ be the median of $\{Y_1,\dots,Y_p\}$. The observation is that the median $\tilde{Y}$ is greater than $\E[Z]+\epsilon$ if and only if more than $p/2$ of means in $\{Y_1,\dots,Y_p\}$ are greater than $\E[Z]+\epsilon$. 

First, we show the probability that $Y_i$ is much larger than $\E[Z]$ is bounded for all $i$. We use the Chebyshev inequality for complex random variables as 
\[
    \Pr[|Y_i - \E[Z]|\geq \epsilon ]\leq \frac{\Var[|Z|]}{\epsilon^2 q}.
\]
Let $q = \frac{4\Var[|Z|]}{\epsilon^2}$ so that the above probability is at most $1/4$. Then, let $E_i$ be the event that $Y_i-\E[Z]>\epsilon$ for $i\in \{1, \ldots, p\}$. By using Chernoff-Hoeffding inequality, we have
\[
    \Pr\left[\sum_{i=1}^p E_i - p\Pr[E_i] \geq p/4\right] \leq e^{-p/8}. 
\]
The probability that the event $\tilde{Y}-\E[Z]>\epsilon$ happens is bounded by
\begin{eqnarray*}
    \Pr[|\tilde{Y}-\E[Z]|\leq\epsilon] &\leq& 
    \Pr[\tilde{Y}-\E[Z]\leq\epsilon]\\
    &=&1 - \Pr[\tilde{Y}-\E[Z]\geq\epsilon]\\
    &=& 1- \Pr\left[\sum_{i=1}^p E_i\geq p/2\right]\\
    &\leq& \Pr\left[\sum_{i=1}^p E_i - p\Pr[E_i] \geq p/4\right]\\
    &\leq& 1- e^{-p/8}.
\end{eqnarray*}
Let $\delta = e^{-p/8}$. By sampling $X$ for a number of $pq = O(\norm{x}\norm{y}\norm{A}_F\frac{1}{\epsilon^2}\log (\frac{1}{\delta}))$ times, dividing them randomly in $p$ groups, and outputting the median of means of these groups, one obtains an estimate of $x^{\dag}Ay$ with additive error at most $\epsilon$ and success probability $1-\delta$.  
    
\end{proof}

The following lemma shows that when vectors $x$ and $y$ are close, the total variation distance between $\mathcal{D}_x$ and $\mathcal{D}_y$ is also close. 
\begin{lemma}\label{lem:dist_distance}
For $x,y\in \mathbb{C}^n$ satisfying $\|x-y\|\leq \epsilon$, it holds that $\|\mathcal{D}_x,\mathcal{D}_y\|_{TV} \leq \frac{2\epsilon}{\|x\|}$. 
\end{lemma}
\begin{proof}
This lemma has been proven in~\cite{T18} when $x,y\in\mathbb{R}^n$. In the following, we show that it is also true for $x,y\in\mathbb{C}^n$.

Let $\bar{x}$ and $\bar{y}$ be the vectors with $\bar{x}(i) = |x(i)|$ and $\bar{y}(i) = |y(i)|$ for all $i \in \{1, \ldots, n\}$. We have

\begin{eqnarray*}
\|\mathcal{D}_x,\mathcal{D}_y\|_{TV} &=& \frac{1}{2} \sum_{i}^n \left|\frac{|x(i)|^2}{\|x\|^2} - \frac{|y(i)|^2}{\|y\|^2}\right| \\
&=& \frac{1}{2}\left\langle \frac{\bar{x}}{\|x\|}-\frac{\bar{y}}{\|y\|}, \frac{\bar{x}}{\|x\|}+\frac{\bar{y}}{\|y\|}\right\rangle\\
&\leq& \frac{1}{2}\left\|\frac{\bar{x}}{\|x\|}-\frac{\bar{y}}{\|y\|}\|\|\frac{\bar{x}}{\|x\|}+\frac{\bar{y}}{\|y\|}\right\|\\
&\leq& \left\|\frac{\bar{x}}{\|x\|}-\frac{\bar{y}}{\|y\|}\right\|\\
&\leq& \frac{1}{\|x\|} \left\|\bar{x} - \bar{y}+ (\|y\|-\|x\|)\frac{\bar{y}}{\|y\|}\right\|\\
&\leq& \frac{2\epsilon}{\|x\|}. 
\end{eqnarray*}
\end{proof}

\section{Sampling a small submatrix}
\label{sec:samping}
In this section, we show a subroutine (Algorithm~\ref{alg:subsampling}) to generate a succinct description of $VD^2V^{\dag}$, which approximates $A^{\dag}A$. This succinct representation allows for efficiently sampling from any column of $V$ as well as querying any entry of it. The intuition of this subroutine is the following: we first obtain a list of real numbers $\hat{\sigma}_1, \ldots, \hat{\sigma}_k \in \bbr$ and a list of vectors $\hat{u}_1, \ldots, \hat{u}_k \in \bbc^p$ from which the matrix $A^{\dag}A$ can be approximately constructed in time $O(nk^2)$; however, we only need the ability to sample from any column of $V$, and hence we can bypass the construction of $V$ to avoid the linear cost (see Section~\ref{sec:main-alg}).

\begin{algorithm}
  \SetKwInOut{Input}{input}\SetKwInOut{Output}{output}
  \Input{$A \in \bbc^{m \times n}$ that satisfies Assumption~\ref{asmp:sample-m}}
  $p \leftarrow 10^7 \cdot \frac{k^{11}\kappa^{20}}{\epsilon^4\|A\|_F^4}$\;
  Independently sample $p$ row indices $i_1, \ldots, i_p$ according to the probability distribution $\{P_1, \ldots, P_n\}$ defined in Assumption~\ref{asmp:sample-m}\;
  Let $S \in \bbc^{p \times n}$ be the matrix formed by the normalized rows $A(i_t, \cdot)/\sqrt{pP_{i_t}}$ for $t \in \{1, \ldots, p\}$\;\label{step:sampling:s}
  Independently sample $p$ column indices $j_1, \ldots, j_p$ by the following procedure: first sample a row index $t \in \{1, \ldots, p\}$ uniformly at random; then sample a column index $j$ from the probability distribution $\{P_1', \ldots, P_n'\}$, where $P_j' = \sum_{t=1}^p \mathcal{D}_{A(i_t,\cdot)}(j)/p$, and $i_1, \ldots, i_p$ are the indices sampled in step 2\;
  Let $W \in \bbc^{p \times p}$ be the matrix formed by the normalized columns $S(\cdot, j_t)/\sqrt{pP_{j_t}'}$ for $t \in \{1, \ldots, p\}$\;\label{step:w}
  Compute the largest $k$ singular values $\hat{\sigma}_1, \ldots, \hat{\sigma}_k$ of $W$ and their corresponding left singular vectors $\hat{u}_1, \ldots, \hat{u}_k$\;
  Output $\hat{\sigma}_1, \ldots, \hat{\sigma}_k$ and $\hat{u}_1, \ldots, \hat{u}_k$\;
  \caption{Subsampling\label{alg:subsampling}}
\end{algorithm}

We claim that the succinct description obtained from Algorithm~\ref{alg:subsampling} can be used to approximate $A^{\dag}A$ in the sense of the following key lemma.
\begin{lemma}\label{lem:appx_A}
  Let $A \in \bbc^{m \times n}$ be a matrix satisfying Assumption~\ref{asmp:sample-m} with $\rank(A) = k$. Take $A$ as the input of Algorithm~\ref{alg:subsampling} and obtain the $\hat{\sigma}_1, \ldots, \hat{\sigma}_k$ and $\hat{u}_1, \ldots, \hat{u}_k$. Let $S$ be the normalized submatrix obtained from step~\ref{step:sampling:s} of Algorithm~\ref{alg:subsampling}. Define $\widehat{A}^{\sim2}$ as
  \begin{align}
    \widehat{A}^{\sim2} = \sum_{i=1}^k \hat{\sigma}_i^2 \frac{S^{\dag}}{\hat{\sigma}_i} \hat{u}_i \hat{u}_i^{\dag} \frac{S}{\hat{\sigma}_i}.
  \end{align}
  Then, with probability at least $9/10$, it holds that $\pnorm{F}{A^{\dag}A - \widehat{A}^{\sim2}} \leq 2\epsilon \pnorm{F}{A}^3/\left(k^{3/2}\kappa^4\right)$.
\end{lemma}

Here, the notation $\widehat{A}^{\sim2}$ suggests that this matrix is close to $A^2$. Before proving this lemma, we need some facts of this algorithm. Let $S$ and $W$ be the normalized submatrices obtained from steps~\ref{step:sampling:s} and \ref{step:w}, respectively, of Algorithm~\ref{alg:subsampling}. First, we need the relationship between $\|S\|_F$ and $\|A\|_F$, and that between $\|W\|_F$ and $\|S\|_F$. This is concluded by the following lemma from \cite{FKV04}, and we do not repeat the proof here.
\begin{lemma}[\cite{FKV04}]
  \label{lemma:asa-sws}
  Let $A \in \bbc^{m \times n}$ satisfying Assumption~\ref{asmp:sample-m}. Take $A$ as the input of Algorithm~\ref{alg:subsampling}, and let $S$ and $W$ be chosen by steps~\ref{step:sampling:s} and \ref{step:w} of Algorithm~\ref{alg:subsampling}. Then, with probability at least $1-16/p$, it holds that
  \begin{align}
    \frac{1}{2}\pnorm{F}{A}^2 \leq \pnorm{F}{S}^2 \leq \frac{3}{2}\pnorm{F}{A}^2, \quad \text{ and } \quad \frac{1}{2}\pnorm{F}{S}^2 \leq \pnorm{F}{W}^2 \leq \frac{3}{2}\pnorm{F}{S}^2.
  \end{align}
\end{lemma}

A key observation is that $S^{\dag}S$ approximates $A^{\dag}A$ and $WW^{\dag}$ approximates $SS^{\dag}$. More precisely, applying Lemma~\ref{lemma:subsample} twice, and using Lemma~\ref{lemma:asa-sws}, we have that if $\theta = \sqrt{40/p}$, then with probability at least $9/10$, the following holds:
\begin{align}
  \pnorm{F}{A^{\dag}A - S^{\dag}S} &\leq \theta\pnorm{F}{A}^2, \text{ and } \\
  \pnorm{F}{SS^{\dag} - WW^{\dag}} &\leq \theta\pnorm{F}{S}^2 \leq \frac{3}{2}\theta\pnorm{F}{A}^2.
\end{align}
In the following analysis, we also need a lower bound on the smallest nonzero singular value of $W$. Without loss of generality, we assume the singular values are in a non-increasing order. Let $\sigma_j(\cdot)$ be the $j$-th singular value of a matrix. To get a lower bound on $\sigma_k(W)$, we choose $\theta = \sqrt{40/p}$. As a consequence of Weyl's inequalities, we have that with probability at least $9/10$, it holds that
\begin{align}
  |\sigma_k(S^{\dag}S) - \sigma_k(A^{\dag}A)| &\leq \norm{A^{\dag}A - S^{\dag}S} \leq \theta\pnorm{F}{A}^2, \text{ and } \\
  |\sigma_k(SS^{\dag}) - \sigma_k(WW^{\dag})| &\leq \norm{SS^{\dag} - WW^{\dag}} \leq \frac{3}{2}\theta\pnorm{F}{A}^2.
\end{align}
Since $\sigma_k(S^{\dag}S) = \sigma_k(SS^{\dag})$, when the error parameter $\epsilon$ in Algorithm~\ref{alg:subsampling} is sufficiently small\footnote{In fact, when $\theta = \sqrt{40/p} = \epsilon^2\pnorm{F}{A}^2/(500k^{11/2}\kappa^{10})$, we need $\epsilon^2 \leq 100k^{11/2}\kappa^8/\pnorm{F}{A}^2$, which is a reasonable assumption.}, we have 
  \begin{align}
    |\sigma_k(WW^{\dag}) - \sigma_k(A^{\dag}A)| \leq \frac{5}{2}\theta\pnorm{F}{A}^2 \leq \frac{\norm{A}^2}{2\kappa^2}.
  \end{align}
Because $\sigma_k(A^{\dag}A) = \norm{A}^2/\kappa^2$, it follows that $\sigma_k(WW^{\dag}) \geq \norm{A}^2/(2\kappa^2)$, and therefore 
\begin{align}
  \label{eq:sigmamin-w}
  \sigma_k(W) \geq \frac{\norm{A}}{\sqrt{2}\kappa}.
\end{align}
With the similar analysis, we can also conclude that with probability at least $9/10$, 
\begin{align}
  \sigma_{1}(W) \leq 2\norm{A},
\end{align}
when $\epsilon$ is sufficiently small.

From Algorithm~\ref{alg:subsampling}, we use the vectors $\frac{S^{\dag}}{\hat{\sigma}_i}\hat{u}_i$ for $i \in \{1, \ldots, k\}$ to approximate the eigenvectors of $A$. In the following lemma, we show that these vectors are almost orthonormal.

\begin{lemma}
  \label{lemma:almost-orthonormal}
  Let $A \in \bbc^{m \times n}$ be a matrix satisfying Assumption~\ref{asmp:sample-m} with $\rank(A) = k$. Take $A$ as the input of Algorithm~\ref{alg:subsampling} and obtain the $\hat{\sigma}_1, \ldots, \hat{\sigma}_k$ and $\hat{u}_1, \ldots, \hat{u}_k$. Let $V \in \bbc^{n \times k}$ be the matrix such that $V(\cdot, j) = \frac{S^{\dag}}{\hat{\sigma}_j}\hat{u}_j$ for $j \in \{1, \ldots, k\}$. Then, with probability at least $9/10$, the following statements hold:
  \begin{enumerate}
    \item There exists an isometry $U \in \bbc^{n \times k}$ whose column vectors span the column space of $V$ satisfying $\pnorm{F}{U - V} \leq \epsilon^2\pnorm{F}{A}^2/(\sqrt{2}k^{7/2}\kappa^8) + O(\epsilon^2)$.
    \item $|\norm{V} - 1| \leq \epsilon^2\pnorm{F}{A}^2/(\sqrt{2}k^{7/2}\kappa^8) + O(\epsilon^2)$.
    \item Let $\Pi_V$ be the projector on the column space of $V$, then it holds that $\pnorm{F}{VV^{\dag} - \Pi_V} \leq \sqrt{2}\epsilon^2\pnorm{F}{A}^2/(k^{7/2}\kappa^8) + O(\epsilon^2)$.
    \item $\pnorm{F}{V^{\dag}V - I} \leq \sqrt{2}\epsilon^2\pnorm{F}{A}^2/(k^{7/2}\kappa^8) + O(\epsilon^2)$.
  \end{enumerate}
\end{lemma}
\begin{proof}
  Most of the arguments in this proof are similar to the proofs in \cite[Lemma 6.6, Corollary 6.7, Proposition 6.11]{T18}. Let $v_j \in \bbc^n$ denote the column vector $V(\cdot, j)$, i.e., $v_j = \frac{S^{\dag}}{\hat{\sigma}_j}\hat{u}_j$. Choose $\theta = \sqrt{40/p} = \epsilon^2\pnorm{F}{A}^2/(500k^{11/2}\kappa^{10})$. When $i \neq j$, with probability at least $9/10$, it holds that
  \begin{align}
    |v_i^{\dag}v_j| = \frac{|\hat{u}_i^{\dag}SS^{\dag}\hat{u}_j|}{\hat{\sigma}_i\hat{\sigma}_j} \leq \frac{|\hat{u}_i^{\dag}(SS^{\dag}-WW^{\dag})\hat{u}_j|}{\hat{\sigma}_i\hat{\sigma}_j} \leq \frac{3\theta\pnorm{F}{A}^2}{2\hat{\sigma}_i\hat{\sigma}_j} \leq \frac{6\theta\kappa^2\pnorm{F}{A}^2}{2\norm{A}^2} \leq \frac{\epsilon^2\pnorm{F}{A}^2}{k^{9/2}\kappa^8},
  \end{align}
  where the second inequality follows from Lemma~\ref{lemma:asa-sws}, and the last inequality uses the fact that $\pnorm{F}{A} \leq \sqrt{k}\norm{A}$. Similarly, when $i = j$, the following holds with probability at least $9/10$.
  \begin{align}
    |\norm{v_i}-1| = \frac{|\hat{u}_i^{\dag}SS^{\dag}\hat{u}_i-\hat{\sigma}_i^2|}{\hat{\sigma}_i^2} \leq \frac{|\hat{u}_i^{\dag}(SS^{\dag}-WW^{\dag})\hat{u}_i|}{\hat{\sigma}_i^2} \leq \frac{3\theta\pnorm{F}{A}^2}{2\hat{\sigma}_i^2} \leq \frac{6\theta\kappa^2\pnorm{F}{A}^2}{2\norm{A}^2} \leq \frac{\epsilon^2\pnorm{F}{A}^2}{k^{9/2}\kappa^8}.
  \end{align}
  Since $|(V^{\dag}V)(i,j)| = |v_i^{\dag}v_j|$, each diagonal entry of $V^{\dag}V$ is at most $\epsilon^2\pnorm{F}{A}^2/(k^{9/2}\kappa^8)$ away from $1$ and each off-diagonal entry is at most $\epsilon^2\pnorm{F}{A}^2/(k^{9/2}\kappa^8)$ away from $0$. More precisely, let $M \in \bbc^{n \times n}$ be the matrix with all ones, i.e., $M(i,j) = 1$ for all $i, j \in \{1, \ldots, n\}$, then for all $i, j \in \{1, \ldots, n\}$, we have
  \begin{align}
    \left(I - \frac{\epsilon^2\pnorm{F}{A}^2}{k^{9/2}\kappa^8} M\right)(i, j) \leq (V^{\dag}V)(i, j) \leq \left(I + \frac{\epsilon^2\pnorm{F}{A}^2}{k^{9/2}\kappa^8} M\right)(i, j).
  \end{align}

  To prove statement 1, we consider the QR decomposition of $V$. Let $Q \in \bbc^{n \times n}$ be a unitary and $R \in \bbc^{n \times k}$ be upper-triangular with positive diagonal entries satisfying $V = QR$. Since $V^{\dag}V = R^{\dag}R$, we have
  \begin{align}
    \left(I - \frac{\epsilon^2\pnorm{F}{A}^2}{k^{9/2}\kappa^8} M\right)(i, j) \leq (R^{\dag}R)(i, j) \leq \left(I + \frac{\epsilon^2\pnorm{F}{A}^2}{k^{9/2}\kappa^8} M\right)(i, j).
  \end{align}
  Let $\widehat{R}$ be the upper $k \times k$ part of $R$. Since $R$ is upper-triangular, $\widehat{R}^{\dag}\widehat{R} = R^{\dag}R$. Hence, $\widehat{R}$ can be viewed as an approximate Cholesky factorization of $I$ with error $\epsilon^2\pnorm{F}{A}^2/(k^{9/2}\kappa^8)$. As a consequence of \cite[Theorem 1]{CPS96}, we have $\pnorm{F}{R - I} \leq \epsilon^2\pnorm{F}{A}^2/(\sqrt{2}k^{7/2}\kappa^8) + O(\epsilon^2)$ for sufficiently small\footnote{Here, it suffices to take $\epsilon \leq k^{7/2}\kappa^8/\pnorm{F}{A}^2$.} $\epsilon$. Now, we define $R' \in \bbc^{n \times k}$ as the matrix with $I$ on the upper $k \times k$ part and zeros everywhere else. Let $U = QR'$. Clearly, $U$ is isometry as it contains the first $k$ columns of $Q$. To see the column vectors of $V$ span the column space of $V$, note that $UU^{\dag} = QR'R'^{\dag}Q^{\dag}$ where $R'R'^{\dag}$ only contains $I$ on its upper-left $(k \times k)$-block, and $VV^{\dag} = QRR^{\dag}Q^{\dag}$ where $R^{\dag}R$ only contains a diagonal matrix on its upper-left $(k \times k)$-block. To bound the distance between $U$ and $V$, we have $\pnorm{F}{U - V} = \pnorm{F}{Q(R' - R)} = \pnorm{F}{R'-R} \leq \epsilon^2\pnorm{F}{A}^2/(\sqrt{2}k^{7/2}\kappa^8) + O(\epsilon^2)$.

  Statement 2 follows from the triangle inequality:
  \begin{align}
    \norm{V} - 1 &= \norm{V} - \norm{U} \leq \norm{V - U} \leq \pnorm{F}{V - U}, \text{ and }\\
    1 - \norm{V} &= \norm{U} - \norm{V} \leq \norm{U - V} \leq \pnorm{F}{U - V}.
  \end{align}

  For statement 3, we have
  \begin{align}
    \pnorm{F}{VV^{\dag} - \Pi_V} &= \pnorm{F}{VV^{\dag} - UU^{\dag}} \\
                                 &\leq \pnorm{F}{V(V^{\dag} - U^{\dag})} + \pnorm{F}{(V-U)U^{\dag}} \\
                                 &\leq \norm{V}\pnorm{F}{V^{\dag} - U^{\dag}} + \pnorm{F}{V-U}\norm{U^{\dag}} \\
                                 &\leq \frac{\sqrt{2}\epsilon^2\pnorm{F}{A}^2}{k^{7/2}\kappa^8} + O(\epsilon^2).
  \end{align}
  Similarly, statement 4 follows by bounding the distance $\pnorm{F}{V^{\dag}V - U^{\dag}U}$.
\end{proof}

Now, we are ready to prove Lemma~\ref{lem:appx_A}
\begin{proof}[Proof of Lemma~\ref{lem:appx_A}]
  Choose $\theta = \sqrt{40/p} = \epsilon^2\pnorm{F}{A}^2/(500k^{11/2}\kappa^{10})$. We use $S^{\dag}WW^{\dag}S$ to approximate $A^{\dag}AA^{\dag}A$. First note that, with probability at least $9/10$, it holds that 
  \begin{align}
    \pnorm{F}{S^{\dag}WW^{\dag}S - S^{\dag}SS^{\dag}S} &\leq \pnorm{F}{S^{\dag}(WW^{\dag} - SS^{\dag})S} \\
                                   &\leq \pnorm{F}{WW^{\dag} - SS^{\dag}}\pnorm{F}{S}^2 \\
    \label{eq:swws-ssss}
                                   &\leq \frac{9}{4}\theta\pnorm{F}{A}^4.
  \end{align}
  We also have that with probability at least $9/10$,
  \begin{align}
    \pnorm{F}{S^{\dag}SS^{\dag}S - A^{\dag}AA^{\dag}A} &\leq \pnorm{F}{S^{\dag}SS^{\dag}S - S^{\dag}SA^{\dag}A + S^{\dag}SA^{\dag}A - A^{\dag}AA^{\dag}A} \\
                                   &\leq \pnorm{F}{S}^2\pnorm{F}{S^{\dag}S - A^{\dag}A} + \pnorm{F}{A}^2\pnorm{F}{S^{\dag}S - A^{\dag}A} \\
    \label{eq:ssss-aaaa}
                                   &\leq \frac{5}{2}\theta\pnorm{F}{A}^4.
  \end{align}
  By Eqns.~\eqref{eq:swws-ssss} and \eqref{eq:ssss-aaaa}, we have
  \begin{align}
    \label{eq:swws-aaaa}
    \pnorm{F}{S^{\dag}WW^{\dag}S - A^{\dag}AA^{\dag}A} \leq \frac{19}{4}\theta\pnorm{F}{A}^4.
  \end{align}
  Now we bound the distance between $(\widehat{A}^{\sim2})^2$ and $S^{\dag}WW^{\dag}S$. First note that
  \begin{align}
    S^{\dag}WW^{\dag}S = \sum_{i=1}^k \hat{\sigma}_i S^{\dag} \hat{u}_i \hat{u}_i^{\dag} S = \sum_{i=1}^k \hat{\sigma}_i^4 \frac{S^{\dag}}{\hat{\sigma}_i} \hat{u}_i\hat{u}_i^{\dag} \frac{S}{\hat{\sigma}_i},
  \end{align}
  where the first equality follows from the fact that $\hat{\sigma}_1, \ldots, \hat{\sigma}_k$ are the singular values of $W$ and $\hat{u}_1, \ldots, \hat{u}_k$ are their corresponding left singular vectors. If the vectors $\frac{S^{\dag}}{\hat{\sigma}_i}\hat{u}_i$ for $i \in \{1, \ldots, k\}$ were orthonormal, $(\widehat{A}^{\sim2})^2 = S^{\dag}WW^{\dag}S$. However, these vectors are approximately orthonormal in the sense of Lemma~\ref{lemma:almost-orthonormal}, which causes difficult. To facilitate the analysis, note that $\widehat{A}^{\sim2}$ can be written as $\widehat{A}^{\sim2} = VD^2V^{\dag}$, where $V(\cdot,i) = \frac{S^{\dag}}{\hat{\sigma}_i}\hat{u}_i$ for $i \in \{1, \ldots, k\}$, and $D$ is the diagonal matrix $D = \Diag(\hat{\sigma}_1, \ldots, \hat{\sigma}_k)$. In the rest of the proof, we focus on bounding the distance between $(\widehat{A}^{\sim2})^2$ and $VD^4V^{\dag}$. 
  
  We first establish the relationship between $\pnorm{F}{\widehat{A}^{\sim2}}$ and $\pnorm{F}{D}^2$. Because of statement 1 of Lemma~\ref{lemma:almost-orthonormal}, there exists a unitary matrix $U \in \bbc^{n \times k}$ whose column vectors span the column space of $V$ such that $\pnorm{F}{V - U} \leq \epsilon^2\pnorm{F}{A}^2/(\sqrt{2}k^{7/2}\kappa^8) + O(\epsilon^2)$. On the one hand, we have
  \begin{align}
    \pnorm{F}{\widehat{A}^{\sim2} - UD^2U^{\dag}} &\leq \pnorm{F}{VD^2(V^{\dag}-U^{\dag})} + \pnorm{F}{(V-U)D^2U^{\dag}} \\
    \label{eq:a-udu-1}
    &\leq \frac{\sqrt{2}\epsilon^2\pnorm{F}{A}^2}{k^{7/2}\kappa^8}\pnorm{F}{D}^2 + O(\epsilon^2).
  \end{align}
  On the other hand, by the triangle inequality,
  \begin{align}
    \label{eq:a-udu-2}
    \pnorm{F}{\widehat{A}^{\sim2} - UD^2U^{\dag}} \geq \pnorm{F}{UD^2U^{\dag}} - \pnorm{F}{\widehat{A}^{\sim2}} = \pnorm{F}{D}^2 - \pnorm{F}{\widehat{A}^{\sim2}}.
  \end{align}
  By Eqns.~\eqref{eq:a-udu-1} and \eqref{eq:a-udu-2}, we have
  \begin{align}
    \pnorm{F}{D}^2 \leq \frac{1}{1-\sqrt{2}\epsilon^2\pnorm{F}{A}^2/(k^{7/2}\kappa^8)}\pnorm{F}{\widehat{A}^{\sim2}} +O(\epsilon^2) \leq (1+2\epsilon^2\pnorm{F}{A}^2/(k^{7/2}\kappa^8))\pnorm{F}{\widehat{A}^{\sim2}} + O(\epsilon^2).
  \end{align}

  Now, we bound the distance between $(\widehat{A}^{\sim2})^2$ and $VD^4V^{\dag}$ as follows:
  \begin{align}
    \pnorm{F}{(\widehat{A}^{\sim2})^2 - VD^4V^{\dag}} &= \pnorm{F}{VD^2V^{\dag}VD^2V^{\dag} - VD^4V^{\dag}} \\
                                                  & = \pnorm{F}{VD^2(V^{\dag}V-I)D^2V^{\dag}} \\
                                                  & \leq \frac{\sqrt{2}\epsilon^2\pnorm{F}{A}^2}{k^{7/2}\kappa^8}\pnorm{F}{D}^4 + O(\epsilon^2) \\
    \label{eq:ahat4-swws}
    &\leq \frac{\sqrt{2}\epsilon^2\pnorm{F}{A}^2}{k^{7/2}\kappa^8}\pnorm{F}{\widehat{A}^{\sim2}}^2 + O(\epsilon^2).
  \end{align}
  Note that $VD^4V^{\dag} = S^{\dag}WW^{\dag}S$. By Eqns.~\eqref{eq:swws-aaaa} and \eqref{eq:ahat4-swws}, we have
  \begin{align}
    \label{eq:ahat4-a4-1}
    \pnorm{F}{(\widehat{A}^{\sim2}) - (A^{\dag}A)^2} \leq \frac{19}{4}\theta\pnorm{F}{A}^4 + \frac{\sqrt{2}\epsilon^2\pnorm{F}{A}^2}{k^{7/2}\kappa^8}\pnorm{F}{\widehat{A}^{\sim2}}^2 + O(\epsilon^2).
  \end{align}
  Now, we also need to examine the distance between $\widehat{A}^{\sim2}$ and $A^{\dag}A$. By the triangle inequality, we have 
  \begin{align}
    \pnorm{F}{(\widehat{A}^{\sim2})^2 - (A^{\dag}A)^2} \geq \pnorm{F}{\widehat{A}^{\sim2}}^2 - \pnorm{F}{A}^4,
  \end{align}
  which implies that
  \begin{align}
    \label{eq:ahat4}
    \pnorm{F}{\widehat{A}^{\sim2}}^2 \leq \frac{1+19\theta/4}{1-\sqrt{2}\epsilon^2\pnorm{F}{A}^2/(k^{7/2}\kappa^8)}\pnorm{F}{A}^4 \leq \left(1+\frac{19}{4}\theta + \frac{2\epsilon^2\pnorm{F}{A}^2}{k^{7/2}\kappa^8}\right)\pnorm{F}{A}^4 + O(\epsilon^2).
  \end{align}
  By Eqns.~\eqref{eq:ahat4-a4-1} and \eqref{eq:ahat4}, we have
  \begin{align}
    \pnorm{F}{(\widehat{A}^{\sim2})^2-(A^{\dag}A)^2} \leq \frac{19}{4}\theta\pnorm{F}{A}^4 + \frac{\sqrt{2}\epsilon^2\pnorm{F}{A}^2}{k^{7/2}\kappa^8}\pnorm{F}{A}^4 + O(\epsilon^4) \leq \frac{\sqrt{2}\epsilon^2}{k^{7/2}\kappa^8}\pnorm{F}{A}^6 + O(\epsilon^2).
  \end{align}
  Finally, applying Lemma~\ref{lemma:x2-y2}, we have
  \begin{align}
    \pnorm{F}{\widehat{A}^{\sim2} - A^{\dag}A} \leq (2k)^{1/4}\pnorm{F}{(\widehat{A}^{\sim2})^2 - (A^{\dag}A)^2}^{1/2} \leq \frac{2\epsilon}{k^{3/2}\kappa^4}\pnorm{F}{A}^3.
  \end{align}
\end{proof}

In Algorithm~\ref{alg:subsampling}, we obtained a succinct description of $\widehat{A}^{\sim2} = \sum_{i=1}^k\hat{\sigma}_i^2\frac{S^{\dag}}{\hat{\sigma}_i}\hat{u}_i\hat{u}_i^{\dag}\frac{S}{\hat{\sigma}_i}$, which is close to $A^{\dag}A$ as a result of Lemma~\ref{lem:appx_A}. For the purpose of the main algorithms, we need to approximate $(A^{\dag}A)^{-1}$. To achieve this, we define the matrix $\widehat{A}^{\sim-2}$ as
\begin{align}
  \widehat{A}^{\sim-2} = \sum_{i=1}^k\hat{\sigma}_i^{-2}\frac{S^{\dag}}{\hat{\sigma}_i}\hat{u}_i\hat{u}_i^{\dag}\frac{S}{\hat{\sigma}_i}.
\end{align}
Note that $\widehat{A}^{\sim-2}$ is not exactly equal to $(\widehat{A}^{\sim2})^{-1}$, as the vectors $\frac{S^{\dag}}{\hat{\sigma}_i}\hat{u}_i$ for $i \in \{1, \ldots, k\}$ are not exactly (but approximately) orthonormal. Another key result of this section is the following lemma

\begin{lemma}
  ~\label{lem:nx}
  Let $A \in \bbc^{m \times n}$ be a matrix satisfying Assumption~\ref{asmp:sample-m} with $\rank(A) = k$. Take $A$ as the input of Algorithm~\ref{alg:subsampling} and obtain the $\hat{\sigma}_1, \ldots, \hat{\sigma}_k$ and $\hat{u}_1, \ldots, \hat{u}_k$. Let $S$ be the normalized submatrix obtained from step~\ref{step:sampling:s} of Algorithm~\ref{alg:subsampling}. Define $\widehat{A}^{\sim-2}$ as
  \begin{align}
    \widehat{A}^{\sim-2} = \sum_{i=1}^k \hat{\sigma}_i^{-2} \frac{S^{\dag}}{\hat{\sigma}_i} \hat{u}_i \hat{u}_i^{\dag} \frac{S}{\hat{\sigma}_i}.
  \end{align}
  Then, with probability at least $9/10$, it holds that $\pnorm{F}{(A^{\dag}A)^{-1} - \widehat{A}^{\sim-2}} \in O(\epsilon)$.
\end{lemma}
\begin{proof}
  We first bound the distance between $(\widehat{A}^{\sim2})^{-1}$ and $(A^{\dag}A)^{-1}$. Observe that $\sigma_{\min}(A^{\dag}A) = \norm{A}^2/\kappa^2$, and with probability at least $9/10$ $\sigma_{\min}(A^{\sim2}) = \sigma_k(A^{\sim2}) \geq \norm{A}^2/(2\kappa^2)$, as shown in Eq.~\eqref{eq:sigmamin-w}. Let $\sigma_{\min} = \max\{\sigma_{\min}(\widehat{A}^{\sim2}), \sigma_{\min}(A^{\dag}A)\} \geq \norm{A}^2/\kappa^2$. By Lemma~\ref{lemma:x-1-y-1}, we have
  \begin{align}
    \pnorm{F}{(\widehat{A}^{\sim2})^{-1}-(A^{\dag}A)^{-1}} &\leq \frac{3\|\widehat{A}^{\sim2}-A^{\dag}A\|_F}{\sigma^2_{\min}} \\
                                                           &\leq \frac{6\epsilon\pnorm{F}{A}^3}{k^{3/2}\kappa^4\sigma_{\min}^2}\\
                                                           &\leq \frac{6\epsilon\pnorm{F}{A}^3}{k^{3/2}\norm{A}^4} \\
    \label{eq:N_v2_1}
    &\leq 6\epsilon,
  \end{align}
  where the last inequality follows from the fact that $\pnorm{F}{A} \leq \sqrt{k}\norm{A}$.

  Next, we bound the distance between $(\widehat{A}^{\sim2})^{-1}$ and $\widehat{A}^{\sim-2}$. As in the proof of Lemma~\ref{lem:appx_A}, define $V \in \bbc^{n \times k}$ as $V(\cdot, i) = \frac{S^{\dag}}{\hat{\sigma}_i}\hat{u}_i$ for $i \in \{1, \ldots, k\}$. Let $\Pi_V$ be the projector onto the space spanned by the column vectors of $V$. By Lemma~\ref{lemma:almost-orthonormal}, there exists an isometry $U$ whose columns vectors span the column space of $V$ satisfying $\pnorm{F}{U-V} \leq \epsilon^2/(\sqrt{2}k^{7/2}\kappa^8) + O(\epsilon^4)$. We have
  \begin{align}
    \pnorm{F}{(\widehat{A}^{\sim2})^{-1} - \widehat{A}^{\sim -2}} &= \pnorm{F}{\left((\widehat{A}^{\sim2})^{-1} - \widehat{A}^{\sim -2}\right)\widehat{A}^{\sim2}(\widehat{A}^{\sim2})^{-1}}\nonumber\\
                                                                &\leq \pnorm{F}{\Pi_V - \widehat{A}^{\sim-2}(\widehat{A}^{\sim2})}\norm{(\widehat{A}^{\sim2})^{-1}} \\
                                                                &\leq \pnorm{F}{ VD^{-2}V^{\dag}VD^{2}V^{\dag}-\Pi_V}\norm{(\widehat{A}^{\sim2})^{-1}} \\
                                                                &\leq \left(\pnorm{F}{VD^{-2}V^{\dag}VD^{2}V^{\dag}-VD^{-2}ID^{2}V^{\dag}}+\pnorm{F}{VV^{\dag}-\Pi_V}\right)\norm{(\widehat{A}^{\sim2})^{-1}} \\
                                                                &\leq \norm{(\widehat{A}^{\sim2})^{-1}}\left(\norm{VD^2}\pnorm{F}{V^{\dag}V-I}\norm{D^{-2}V}+\frac{\epsilon^2}{k^{7/2}\kappa^4}\right) + O(\epsilon^4)\\
                                                                &\leq \norm{(\widehat{A}^{\sim2})^{-1}}\left(\norm{D}^2\norm{D^{-1}}^2\frac{\epsilon^2}{\sqrt{2}k^{7/2}\kappa^4}+\frac{\epsilon^2}{k^{7/2}\kappa^4}\right) + O(\epsilon^4)\\
                                                                &\leq \norm{(\widehat{A}^{\sim2})^{-1}}\left(\frac{(4\kappa^2+1)\epsilon^2}{k^{7/2}\kappa^4}\right) + O(\epsilon^4)\\
                                                                &\leq \frac{(8\kappa^4+2\kappa^2)\epsilon^2}{k^{7/2}\kappa^4\norm{A}^2} + O(\epsilon^4).
  \end{align} 
  Finally, we have
  \begin{align}
    \pnorm{F}{\widehat{A}^{\sim-2}-(A^{\dag}A)^{-1}} &\leq \pnorm{F}{\widehat{A}^{\sim-2} - (\widehat{A}^{\sim2})^{-1}}+\pnorm{F}{(\widehat{A}^{\sim2})^{-1}-(A^{\dag}A)^{-1}}\\
                                                     &\leq 6\epsilon + \frac{(8\kappa^4+2\kappa^2)\epsilon^2}{k^{7/2}\kappa^4\norm{A}^2} + O(\epsilon^4) \\ 
    \label{eq:n31}
                                                     &\in O(\epsilon).
   \end{align}
\end{proof}

\section{Main algorithm and proofs of main theorems}
\label{sec:main-alg}
In this section, we present the main algorithms and prove the main theorems.

Before showing the main algorithms, we briefly give the idea. Instead of considering $A^{-1} b$ directly, our algorithms aim to implement $(A^{\dag}A)^{-1}A^{\dag}b$. The following claim shows that $(A^{\dag}A)^{-1}A^{\dag} b = A^{-1} b$ for $A\in \mathbb{C}^{m\times n}$ and $b \in \bbc^m$.  
\begin{claim}~\label{claim:N_1}
Let $M\in \mathbb{C}^{m\times n}$ and $v\in \mathbb{C}^m$. Then, $M^{-1}v = (M^{\dag}M)^{-1}M^{\dag}v$. 
\end{claim}
\begin{proof}
Let $UDV^{\dag} = M$ be the singular decomposition of $M$. Then, 
\[
    (M^{\dag}M)^{-1}M^{\dag}v = (VDU^\dag UDV^{\dag})^{-1}VDU^\dag v = VD^{-1}U^\dag v= M^{-1} v.  
\]
\end{proof}

Given matrix $A$ and vector $b$, we work on the matrix $\widehat{A}^{\sim-2} \in \mathbb{C}^{n\times n}$ obtained by Algorithm~\ref{alg:subsampling} as $\|(\widehat{A}^{\sim-2} - (A^\dag A)^{-1}\|$ is bounded according to Lemma~\ref{lem:nx}. Then we use the sampling techniques in Section~\ref{sec:samp_tech} to accomplish the task of sampling and querying $\widehat{A}^{\sim-2}A^\dag b$. 

\begin{algorithm}
  \SetKwInOut{Input}{input}\SetKwInOut{Output}{output}
  \Input{$A \in \bbc^{m \times n}$ with access as in Assumption~\ref{asmp:sample-m}, $b \in \bbc^m$ with the sampling access as in Assumption~\ref{asmp:sample-v}, and an index $j \in \{1, \ldots, n\}$}
  Take $A$ as the input of Algorithm~\ref{alg:subsampling} and obtain a description of $\widehat{A}^{\sim2} = VD^2V^{\dag}$, where $V(\cdot, i) = \frac{S^T\hat{u}_i}{\hat{\sigma}_i}$ for $i \in \{1, \ldots, k\}$ and $D = \Diag(\hat{\sigma}_1, \ldots, \hat{\sigma}_k)$\;
  Use Lemma~\ref{lem:samp_4} to estimate $V^{\dag}(i, \cdot)A^{\dag}b$ for $i \in \{1, \ldots, k\}$. The query and sampling access to $V^{\dag}$ is established in Lemma~\ref{lem:samp_2}\;
  Construct the vector $w \in \bbc^{k}$ as $w(i) = V^{\dag}(i, \cdot)A^{\dag}b$\;
  Compute $w' = D^{-2}w$\;
  Output the inner product $V(j, \cdot)w'$\;
  \caption{Solving Linear system for querying an entry of $A^{-1}b$ \label{alg:query}}
\end{algorithm}

\mainquery*

\begin{proof}

The algorithm we use is shown in Algorithm~\ref{alg:query}. We first prove the correctness as follows. In step 1, one can compute $V(i,j)$ and sample each element in $V(\cdot,j)$ by using the data structure (as shown in Theorem~\ref{thm:data-m}) and Lemma~\ref{lem:samp_2}. Step 2 yields values with additive $\tilde{\epsilon}'$ to $V^\dag(i,\cdot) A^{\dag} b$ by Lemma~\ref{lem:samp_4}. Note that $w$ is not exactly $V^{\dag}A^{\dag}b$, but an estimate. Hence, 
\[
   \|w-V^\dag A^{\dag}b\|\leq \tilde{\epsilon}' \sqrt{k}.
\]
Steps 4 and 5 can be done without error since evaluating the inner product only takes constant time.  Therefore, the sampling error is bounded by
\begin{align}
\label{eq:vw'1}
    \|Vw' - V D^{-2} V^\dag Ab\| \leq 
    \|V\|\|D^{-2}\|\|w-V^\dag Ab\|
    \leq O\left( \frac{\tilde{\epsilon}'\sqrt{k}\kappa^2}{\|A\|^2}\right). 
\end{align}
Then, by Lemma~\ref{lem:nx} (with error parameter $\tilde{\epsilon}$ in Algorithm~\ref{alg:subsampling}), the approximation error is bounded by
\begin{align}
    \|(\widehat{A}^{\sim-2}A^\dag-(A^\dag A)^{-1}A^\dag)b \|&\leq 
    \|(\widehat{A}^{\sim-2}-(A^\dag A)^{-1}\|\|A\|\|b\|\\
    \label{eq:vw'2}
    &\leq O(\tilde{\epsilon} \|A\|\|b\|). 
\end{align}
The output of the algorithm has an error $O(\tilde{\epsilon}' \sqrt{k}\kappa^2/\|A\|^2 +\tilde{\epsilon}\|A\|\|b\|)$, where $\tilde{\epsilon}'$ is the sampling error from Lemma~\ref{lem:samp_4} and $\tilde{\epsilon}$ is the approximation error from Lemma~\ref{lem:nx}. To bound the total error by a single $\epsilon$, we rescale the error parameter in step 1 of Algorithm~\ref{alg:subsampling} by $\tilde{\epsilon} = \epsilon/(\norm{A}\norm{b})$, and rescale the error parameter in Lemma~\ref{lem:samp_4} by $\tilde{\epsilon}' = \epsilon\norm{A}^2/(\sqrt{k}\kappa^2)$. Hence, in step 1 of Algorithm~\ref{alg:subsampling}, $p$ becomes $p = 10^7\cdot\frac{k^{11}\kappa^{20}\norm{b}^4}{\epsilon^4}$. The time complexity of Algorithm~\ref{alg:query} is dominated by $O(p^3) = O(\frac{k^{33}\kappa^{60}\norm{b}^{12}}{\epsilon^{12}})$ for applying singular value decomposition. 

The query complexity can be calculated as follows. According to Lemma~\ref{lem:samp_2}, given $i\in \{1, \ldots, n\}$, one can output a sample from $V(\cdot,i)$ with probability $9/10$ in $O(p^2C(S^{\dag},\hat{u}_i/\hat{\sigma}_i))$ queries and output any entry of $V(\cdot,i)$ in $O(p)$ queries. Here,
\[
   C(S^{\dag},\hat{u}_i/\hat{\sigma}_i)) = \frac{\sum_{j=1}^p \|\hat{u}_i(j)S(j,\cdot)/\hat{\sigma}_i\|^2}{\|S^{\dag}\hat{u}_i/\hat{\sigma}_i\|^2}
   \leq\frac{(3/2)\|\hat{u}_i/\hat{\sigma}_i\|^2\|A\|_F^2}{\left(1-\frac{\epsilon^2\|A\|^2_F}{\sqrt{2}k^{7/2}\kappa^8}-O(\epsilon^2)\right)^2}
   \leq O(\kappa^2k). 
\]
The first inequality is by Lemma~\ref{lemma:almost-orthonormal}, where $\|S^{\dag}\hat{u}_i/\hat{\sigma}_i\| \geq 1- \epsilon^2\|A\|^2_F/(k^{7/2}\kappa^8)-O(\epsilon^2)$.  The last inequality is true since the minimum nonzero singular value of $W$ is at least $\|A\|^2/2\kappa$, which is shown by Eq.~\eqref{eq:sigmamin-w}. Then, step 3 requires $O(\frac{k^2\norm{A}_F\norm{V(\cdot,j)}\norm{b}\kappa^4}{\epsilon^2\norm{A}^4}\log(\frac{1}{\delta}))$ query complexity to output the vector $w$. (Note that we have rescaled the error parameter for Lemma~\ref{lem:samp_4}.) Given $w$, doing the calculation in step 4 and step 5 requires $O(k^2)$ steps. Hence the algorithm uses 
\begin{align}
\label{eq:query-complexity}
    O\left(p^2\left(\kappa^2k\right)\frac{k^2\pnorm{F}{A}\norm{V(\cdot,j)}\norm{b}\kappa^4}{\epsilon ^2\norm{A}^4}\polylog(m,n)\log\frac{1}{\delta}\right)
    =O\left(\frac{k^{26}\kappa^{46}\|b\|^9}{\epsilon^{10}\|A\|^{3}}\polylog(m,n)\log\frac{1}{\delta}\right)
\end{align}
queries and runs in time 
\begin{align}
\label{eq:time-complexity}
    O\left(\max\left\{\frac{k^{33}\kappa^{60}\norm{b}^{12}}{\epsilon^{12}},\frac{k^{26}\kappa^{46}\|b\|^9}{\epsilon^{10}\|A\|^{3}}\right\}\polylog(m,n)\log\frac{1}{\delta}\right)
\end{align}
to succeed with probability $1-\delta$.  
\end{proof}

\begin{algorithm}
  \SetKwInOut{Input}{input}\SetKwInOut{Output}{output}
  \Input{$A \in \bbc^{m \times n}$ with the sampling access as in Assumption~\ref{asmp:sample-m} and $b \in \bbc^m$ with the sampling access as in Assumption~\ref{asmp:sample-v}}
  Take $A$ as the input of Algorithm~\ref{alg:subsampling} and obtain a description of $\widehat{A}^{\sim2} = VD^2V^{\dag}$, where $V(\cdot, i) = \frac{S^T\hat{u}_i}{\hat{\sigma}_i}$ for $i \in \{1, \ldots, k\}$ and $D = \Diag(\hat{\sigma}_1, \ldots, \hat{\sigma}_k)$\;
  Use Lemma~\ref{lem:samp_4} to estimate $V^{\dag}(i, \cdot)A^{\dag}b$ for $i \in \{1, \ldots, k\}$. The query and sampling access to $V^{\dag}$ is established in Lemma~\ref{lem:samp_2}\;
  Construct the vector $w \in \bbc^{k}$ as $w(i) = V^{\dag}(i,\cdot)A^{\dag}b$\;
  Compute $w' = D^{-2}w$\;
  Use Lemma~\ref{lem:samp_3} to sample $Vw'$. The query and sampling access to $V$ is established in Lemma~\ref{lem:samp_2}\;
  \caption{Solving Linear system for sampling from $A^{-1}b$ \label{alg:samp}}
\end{algorithm}

\mainsample*

\begin{proof}
The proof mostly follows the proof of Theorem~\ref{thm:main1}. The main difference is that the last step of Algorithm~\ref{alg:samp} requires to output a sample from $\mathcal{D}_{V w'}$. This is done by using Lemma~\ref{lem:samp_3}. Let $D_{w''}$ be the distribution that Algorithm~\ref{alg:samp} actually samples from in step 5. Algorithm~\ref{alg:samp} guarantees to give a sample from $\mathcal{D}_{w''}$ which is $O(\frac{\tilde{\epsilon}^2\|A\|^2_F}{k^{7/2}\kappa^{8}})$-close to the distribution of $\mathcal{D}_{Vw'}$ in terms of total variation distance. Note that we use $\tilde{\epsilon}$ for the error parameter in step 1 of Algorithm~\ref{alg:subsampling} and use $\tilde{\epsilon}'$ for the error parameter for Lemma~\ref{lem:samp_4}. Together with the fact that $\|\widehat{A}^{\sim-2}-(A^\dag A)^{-1}\|\leq O(\tilde{\epsilon})$, we have
\begin{align*}
    \|\mathcal{D}_{w''},\mathcal{D}_{(A^\dag A)^{-1}A^\dag b}\|_{TV} &\leq \|\mathcal{D}_{w''},\mathcal{D}_{Vw'}\|_{TV}+\|\mathcal{D}_{Vw'},\mathcal{D}_{(A^\dag A)^{-1}A^\dag b}\|_{TV} \\
    &\leq O\left(\frac{\tilde{\epsilon}^2\|A\|^2_F}{k^{7/2}\kappa^{8}}\right) + O\left(\tilde{\epsilon}' \sqrt{k}\kappa^2/\norm{A}^2 + \tilde{\epsilon}\|A\|\|b\|\right) \\
    &\leq O(\tilde{\epsilon}) + O\left(\tilde{\epsilon}' \sqrt{k}\kappa^2/\norm{A}^2 + \tilde{\epsilon}\|A\|\|b\|\right), 
\end{align*}
where the second big-$Q$ follows from Eqns.~\eqref{eq:vw'1} and \eqref{eq:vw'2}. Again, we rescale the error parameters as $\tilde{\epsilon} = \epsilon/(\norm{A}\norm{b})$ and $\tilde{\epsilon}' = \epsilon\norm{A}^2/(\sqrt{k}\kappa^2)$ so that the total error is $\epsilon$.

The analysis of time and query complexity is the same as the proof of Theorem~\ref{thm:main1} except that the last step for sampling from $\mathcal{D}_{Vw'}$ requires $O\left(k^2\left(1+O\left(\frac{\tilde{\epsilon}^2\|A\|^2_F}{k^{7/2}\kappa^{8}}\right)\right)\right)$ query and time complexity, which is dominated by other terms in Eqns.~\eqref{eq:query-complexity} and \eqref{eq:time-complexity}.

\end{proof}

\subsection{When \texorpdfstring{$b$}{b} is not exactly in the left-singular vector space of \texorpdfstring{$A$}{A}}
\label{subsec:b-not-in-a}

We have mentioned in the beginning that when $b$ has little or zero overlaps with the left singular vector space of $A$, the outputs of our algorithms may be dominated by the additive error from sampling and approximation. More specifically, let $S_\ell$ be the left singular vector space of $A$. In the case where $b = \sqrt{c} b_A + \sqrt{1-c} \bar{b}_A$ for $b_A\in S_\ell$, $\bar{b}_A\notin S_\ell$, and $c \ll 1$, it is very likely that $\epsilon \gg b_{A}(i)$ for most $i$. Therefore, the sampling subroutine in step 5 of Algorithm~\ref{alg:samp} is not reasonable. 

However, we can estimate the projection of $b$ on the left-singular vector space of $A$ by evaluating $b^{\dag} A (VD^{-2}V^\dag A^\dag)\approx b^{\dag} A^{-1}b$ according to Lemma~\ref{lem:samp_4}. With this routine, one can set a threshold such that only $b$'s which projections on the left-singular vector space of $A$ are greater than the threshold will be considered to be sampled.

\subsection{When the sampling access to \texorpdfstring{$b$}{b} is not given}
\label{subsec:no-sampling-b}
In this subsection, we sketch an algorithm for solving the problem when the sampling assumption of $b$ is not given and $A$ is positive semidefinite. Recall that in Algorithm~\ref{alg:subsampling}, we have obtained a $\hat{\sigma}_1, \ldots, \hat{\sigma}_k \in \bbr$ and $\hat{u}_1, \ldots, \hat{u}_k \in \bbc^{p}$. Again, define the diagonal matrix $D$ as $D = \Diag(\hat{\sigma}_1, \ldots, \hat{\sigma}_k)$ and matrix $V$ as $V(\cdot, j) = \frac{S^{\dag}}{\hat{\sigma}_j}\hat{u}_j$, where $S$ is the normalized submatrix sampled from step~\ref{step:sampling:s} in Algorithm~\ref{alg:subsampling}. In Lemma~\ref{lem:appx_A}, we have shown that $VD^2V^{\dag} \approx A^{\dag}A$. If $A$ is positive semidefinite, by using Lemma~\ref{lemma:x2-y2} again (and by arguing $VD^2V^{\dag} \approx (VDV^{\dag})^2$ because $V$ is close to an isometry as shown in Lemma~\ref{lemma:almost-orthonormal}), one can show that $VDV^{\dag} \approx A$. Then, invoking Lemma~\ref{lemma:x-1-y-1} (and by arguing $VD^{-1}V^{\dag} \approx (VDV^{\dag})^{-1}$), it can be shown that $VD^{-1}V^{\dag} \approx A^{-1}$. Now, we use Lemma~\ref{lem:samp_1} $k$ times to get an estimate of the $V^{\dag}b$ (where no sampling assumption for $b$ is required). The vector $D^{-1}V^{\dag}b$ can be computed easily. Finally, we use Lemma~\ref{lem:samp_3} to sample from and query to $V(D^{-1}V^{\dag}b)$.

\section{Acknowledgements}
We thank Scott Aaronson for the valuable feedback on a draft of this paper. We also thank Tongyang Li for pointing out reference~\cite{AKP18}. 

%\printbibliography
\bibliographystyle{plain}
\bibliography{ref.bib}
\appendix
\end{document}